\documentclass{article}  
                                                  
\usepackage{authblk}

\title{\LARGE \bf
Output Regulation for Systems on Matrix Lie-group
}
\date{}
\author[1]{Simone de Marco \thanks{simone.demarco2@unibo.it}}
\author[1]{Lorenzo Marconi \thanks{lorenzo.marconi@unibo.it}}
\author[2]{Tarek Hamel\thanks{thamel@i3s.unice.fr}}
\author[3]{Robert Mahony\thanks{Robert.Mahony@anu.edu.au}}
\affil[1]{CASY-DEI, Universit\'a di Bologna, Bologna, 40133, Italy}
\affil[2]{I3S, University of Nice, Sophia Antipolis, France}
\affil[3]{Centre of Excellence in Robotic Vision, Research School of Engineering, Australian National University, Canberra, Australia}

\usepackage{amsfonts,latexsym,amsmath,amssymb}
\usepackage{amsthm}
\usepackage[mathscr]{eucal}
\usepackage{mathtools}
\usepackage{graphicx}              
\usepackage{float}
\usepackage{color}    
\usepackage[export]{adjustbox}                                               
\usepackage[dvips]{epsfig}    
                                               
\newcommand{\ssum}[2]{\overset{#2}{\underset{#1}{\sum}}}   
\newcommand{\vect}{\mbox{\rm vec}}
\newcommand{\RR}{\mathbb R}
\newcommand{\dnorm}[1]{\Vert #1 \Vert}
\newcommand{\norm}[1]{\left|#1\right|}

\def\be{\begin{equation}}
\def\eeq{\end{equation}}
\def\ba{\begin{array}}
\def\ea{\end{array}}       
         
\newtheorem{thm}{Theorem}
\newtheorem{rem}{Remark}
\newtheorem{prop}{Proposition}  
\newtheorem{assum}{Assumption}    
                          
\newcommand{\SO}{\mathrm{SO}}
\newcommand{\SL}{\mathrm{SL}}

\newcommand{\SE}{\mathrm{SE}}

\newcommand{\so}{\mathfrak{so}}

\newcommand{\gothg}{\mathfrak{g}}

\newcommand{\grpG}{\mathbf{G}}

\newcommand{\R}{\mathbb{R}}

\newcommand{\Id}{I}

\newcommand{\EE}{E}
\newcommand{\ee}{e}

\newcommand{\cost}{\ell} 

\newcommand{\pr}{\mathbb{P}}

\DeclareMathOperator{\vrp}{vrp} 
\DeclareMathOperator{\mrp}{mrp} 
\newcommand{\Vrp}{\mathfrak{v}} 

\DeclareMathOperator{\tr}{tr}
\newcommand{\trace}[1]{\tr\left(#1\right)}

\DeclareMathOperator{\diag}{diag}

\DeclareMathOperator{\Ad}{Ad}

\mathchardef\mhyphen="2D  

\begin{document}
\maketitle
\thispagestyle{empty}
\pagestyle{empty}

\begin{abstract}                          
\textbf{This paper deals with the problem of output regulation for systems defined on matrix Lie-Groups. Reference trajectories to be tracked are supposed to be generated by an exosystem, defined on the same Lie-Group of the controlled system, and only partial relative error measurements are supposed to be available. These measurements are assumed to be invariant and associated to a group action on a homogeneous space of the state space. In the spirit of the internal model principle the proposed control structure embeds a copy of the exosystem kinematic. This control problem is motivated by many real applications fields in aerospace, robotics, projective geometry, to name a few, in which systems are defined on matrix Lie-groups and references in the associated homogenous spaces.}
\end{abstract}

\section{Introduction}
Regulating the output of a system in such a way to achieve asymptotic tracking of a reference trajectory is a central topic in control. Among the different approaches proposed so far  in the related  literature, the output regulation undoubtedly plays a central role. The main peculiarity of output regulation is the fact of considering references to be tracked as belonging to a family of trajectories generated as solutions of  an autonomous system (typically referred to as exosystem). Tackling this problem in error feedback contexts typically leads to solutions in which the regulator embeds a copy of the exosystem properly updated by means of error measurements.
 The problem has been intensively studied in the linear context in the mid seventies by Davison, Francis and Wonham \cite{FrancisWonham},\cite{Francis},\cite{Davison}, and then extended to a quite general nonlinear context by Isidori and Byrnes \cite{Isidori}. In both the linear and nonlinear framework the developed theory has led to the celebrated internal model principle. Assumptions of the aforementioned pioneering works were weakened recently \cite{isidoriHighGain},  \cite{MarconiAdaptiveObserver}  thanks to the observation that the problem of output regulation can be cast as the problem of nonlinear observers design. In fact, the many tools available in literature  for designing nonlinear observers, such as high-gain and adaptive nonlinear observers, have been used in the design of internal model-based regulators by thus widening the class of systems that can dealt with. It is worth noting, however, that most of the frameworks considered so far for output regulation deal with systems and exosytems defined on real state space and not much efforts have been done to extend the results of output regulation  to systems and exosystems whose states live in more general manifolds. 

Regarding control and observation of systems defined in more general manifolds, the synthesis of nonlinear observers for invariant systems on Lie-groups has recently played an important role.  Many physical systems, such as aerial vehicles, mobile robotic vehicles, robotic manipulators, can be described by geometric models with symmetries. Symmetric structures reflect the fact that the behavior of a symmetric system in one point in the space is independent from the choice of a  particular set of configuration coordinates. That is, the laws of motion of a symmetric system are invariant under a change of the configuration space. Preserving such a symmetry is undoubtedly a key point in the observer design. Aghannan and Rouchon \cite{Aghannan} first had pointed out the main role of invariance in the observer design for mechanical systems with symmetries. More recent works, based on the aforementioned paper, take advantage of the symmetry of left invariant systems on Lie-groups to define invariant error coordinates in order to build an invariant observer \cite{Bonnabel}, \cite{BonnabelLie}, \cite{MinhDuc}, \cite{Trumpf2012}. Lageman et. al  \cite{Lageman2009}, \cite{Lageman2010} exploit the invariance properties of the system in order to have autonomous invariant error dynamics. 

This paper presents  an attempt to extend the idea of internal model-based control to systems defined on more general state spaces, precisely systems defined on matrix Lie-Groups. 
This control problem is motivated by a wide range of real world applications in which both the controlled system and the exosystem are modeled on matrix Lie-groups. For example, the attitude control problem for satellite \cite{BulloPHD}, in which the system and the exosystem are defined on the group of rigid-body rotations $\SO(3)$ or the control \cite{MarconiRobustGuidance} of VTOL (Vertical Takeoff and Landing), whose kinematic equations are described by the group of rigid-body translation $\SE(3)$.
The relevance of this control problem is motivated also by projective geometry problem such as image homographies.
This geometry can be modeled by the special linear group $\SL (3)$ \cite{YMa}, \cite{Mahony2012} and in this context the control problem is to continuously warp a reference image of a reference scene in such a way the  ``controlled" homography matches an image sequence taken by a moving camera.
 
In this paper, we consider reference trajectories to be tracked modelled on the same Lie-group of the controlled system. The exosystem (the reference system) is assumed to be right invariant while the controlled system is assumed to be left invariant. We assume that only partial relative error information are available for measurements. The measurements are assumed to be associated with a left group action on a homogeneous space of the state space. 
 The paper extends preliminary results presented in \cite{ACC15} where the problem of output regulation for systems defined on matrix Lie-Group was tackled under the assumption that the exosystem has constant velocity. 
 Here we remove the limitations of \cite{ACC15} by considering a generic exosystem structure. The main result of the  paper is to propose a general structure of the regulator that depends only on relative measures in homogenous space such that the closed-loop system tracks homogenous references generated by the exosystem with a certain domain of attraction.  This is achieved by means of a regulator embedding a copy of the exosystem fed by a non-linear function of the measured output  along with a stabilising control action. Going further, a control design based on \emph{back-stepping} techniques, in the special case of systems defined on the special orthogonal group $\SO(3)$, is presented  for the fully actuated dynamic model. 
 
The paper has six sections and it's organised as follow. Section \ref{NotationAndProblemFormulation} presents mathematical preliminaries, notation and the problem formulation. Section \ref{OutputRegulation} presents the main result of the paper, a theorem for kinematic systems on matrix Lie-Groups with invariant relative error measurements and its proof. In Section \ref{SO3case} a stability analysis is provided for the specific case of systems defined on the \emph{special orthogonal group} extending local results of Section \ref{OutputRegulation} to almost global results. In the same section is also proposed a regulator design, based on \emph{back-stepping} techniques, for fully actuated dynamic systems whose kinematic state space is posed on the Lie-Group of orthogonal rotations $\SO(3)$. An illustrative example with simulations is shown in Section \ref{Simulation Results}.
 
 \section{Notation and Problem Formulation} \label{NotationAndProblemFormulation}
\subsection{Notation and basic facts}
Let $\grpG$ be a general matrix Lie-group. For $X\in\grpG$, the group inverse element is denoted by $X^{-1}$, and $I$ is the identity element of $\grpG$. Let $\gothg$ be the Lie-algebra associated to the Lie-group $\grpG$.
For $X\in \grpG$ and $U\in \gothg$, the $\Ad_X U$ denotes the adjoint operator, namely the mapping $\Ad_X:\grpG \times \gothg \rightarrow \gothg$ defined as
$$\Ad_X U:=XUX^{-1},\quad  \text{with} \>\, X \in \grpG, \>U \in \gothg \,.$$
Let $\trace{A}$ denote the trace operator and note that for any $A,B \, \in \R^{n\times n}$,  $\trace{A^\top B}$ defines an inner product on $\R^{n \times n}$.
For any $A \in \R^{n \times n}$ and $B \in \R^{n \times n}$
$$\trace{A^\top B}=\vect (A)^\top\vect(B)$$
where $\vect(A) \in \R^{n^2}$ is the column vector obtained by the concatenation of columns of the matrix $A$ as follows
$$\vect(A)=[a_{1,1},...,a_{n,1},a_{1,2},...,a_{n,2},...,a_{1,n},...,a_{n,n}]^{\top}.$$
For all $A \in \R^{n \times n}$  let $\pr_{\gothg}(A)$ denote the orthogonal projection of $A$ onto $ \gothg $ with respect to the trace inner product. For any $U \in \gothg$ and any $A \in \R^{n \times n}$, one has
$$\trace{U^\top A}=\trace{U^\top \pr_\gothg(A)}\, .$$ 
We denote by $\mrp$ (matrix representation) the mapping $\mrp: \R^n \rightarrow \gothg$ that maps the vector $v \in \R^n$ in an element of the algebra.
Let $\vrp$ (vector representation) denote the inverse of $\mrp$ operator, namely
\[\ba{rclrcl}
\vrp(\mrp(v))&=&\> \>v,\quad &v&  \in \R^n \\
\vrp(U) &=& U_\Vrp,\quad &U& \in \gothg \, .
\ea\]
With $\Omega \in \RR^3$ we denote by $\Omega_\times$ the matrix representation of $\Omega$ in the Lie-algebra $\so(3)$ associated to the Lie-group $\SO(3)$. Namely, 
 if $\Omega=[\Omega_1, \Omega_2, \Omega_3]^\top$, then 
$$
\Omega_\times=
\begin{bmatrix}
0                  &    -\Omega_3  &   \Omega_2\\
\Omega_3    &                   0  &  -\Omega_1\\
-\Omega_2   &      \Omega_1 &                 0
\end{bmatrix}.
$$
Let $U \in \gothg$ be an $n \times n$ matrix and let $U_\Vrp \in \R^k$ with $n^2 \geq k$. Then, it is straightforward to verify that $\mbox{vec}(U)$ is a linear combination of the vectorial representation $\vrp(U)$ and there exists a duplication matrix $D \in \R^{n^2 \times k}$ such that
$$\vect(U)=DU_\Vrp \, .$$
With $U \in \gothg$ and $V \in \gothg$ elements of the same Lie-algebra one has
$$\trace{U^\top V}=\vrp^\top(U) Q_{\gothg} \vrp (V)$$
where $Q_{\gothg}=D^T D$ and $D$ is the duplication matrix defined above.

\subsection{Problem formulation}
In this paper we consider the left invariant kinematic system
\be\label{System}
\dot{X}=XU,	\quad X(0) \in \grpG
\eeq
$X \in \grpG$ the state of the system and $U \in \gothg$ the control input.
We consider a right invariant exosystem on $\textbf{G}$  of the form
\begin{subequations}\label{ExoSystem}
\be
  \dot X_d={}^\circ U_d X_d
\eeq
\be
\quad {}^\circ U_d =\mrp(Cw)
\eeq
\be
\!\!\!\! \dot w= Sw
\eeq
\end{subequations}
where $X_d \in \grpG$ and ${}^\circ U_d \in\gothg$ are $n \times n$ matrices, $\vrp({}^\circ U_d) \in \R^\kappa$,$w \in \R^m$, $C \in \R^{\kappa \times m}$, and  $S \in \R^{m\times m}$ with $m \geq \kappa$ and $S=-S^{\top}$.
The input $w$ models exogenus signals that represent velocity references to be tracked.
Note that $U \in \gothg$ in \eqref{System} is associated to a body-fixed velocity while ${}^\circ U_d \in \gothg$ of the exosystem \eqref{ExoSystem} is associated with an inertial velocity input.
Consider a linear left group action\footnote{In the paper  a linear left group action is considered since it is the natural action for the Lie-groups $\SO(n)$ and $\SE(n)$.  All results presented, however,  hold also for a generic left group action.} of $\textbf{G}$ on $\RR^n$, $(X,y)\mapsto Xy$ and
reference vectors in the form
$$
y_i=X^{-1} X_d \mathring y_i, \qquad i=1,2,...,\nu
$$
where $\mathring y_i$ are known constant reference vectors, elements of the homogenous space associated to the Lie-group $X$.
Let
\be\label{stateError}
\EE=X^{-1}X_d
\eeq
denotes the state error of the system as an element of the group $\grpG$.
Let $X_r \in \grpG$ be a constant â ``reference" element of the group.  $X_r$ can be applied to $\mathring y_i$ to generate constant reference vectors
$$
y_i^r=X_r \mathring y_i \> .
$$
Define an â ``error vector" $e_i$ by
\be
\ee_i=y_i^r-y_i=y_i^r-E_r y_i^r, \quad \mbox{\rm where} \quad E_r=EX_r^{-1} \> .
\eeq
The control problem considered is the design of a feedback control action $U$ as a function of $X$ and $y_i$, in such a way the error $e_i$ converges to zero with a certain domain of attraction.

To have a physical intuition of those reference  and error vectors we address the specific case of $\SE (2)$. In particular (see Figure \ref{reference vectors}), we consider two mobile robots modeled as kinematic systems. The first one, the exosystem, is moving along a certain trajectory while the controlled vehicle measures its relative position with respect to the exosystem.
In this context $\mathring y_i$  is a known point in the frame associated to the exosystem mobile robot (for example the coordinate of  $\mathring y_1=[1, 0, 1]^\top$) in its own body-fixed frame with an offset given by $X_r$.
The known vector $X_d \mathring y_1$ is the inertial representation of $\mathring y_1$, while $\EE_r \mathring y_1$ expresses the vector $\mathring y_1$ in the frame associated to the actual robot.
In this context the control objective is to drive the actual mobile robot with the velocity as control input in order to track the exosystem (the desired one) in such a way the error  $y_1^r-\EE_r y_1^r$ converges to zero with a certain domain of attraction.

\begin{figure}[h]
\begin{center}
    \includegraphics[width=.45\textwidth]{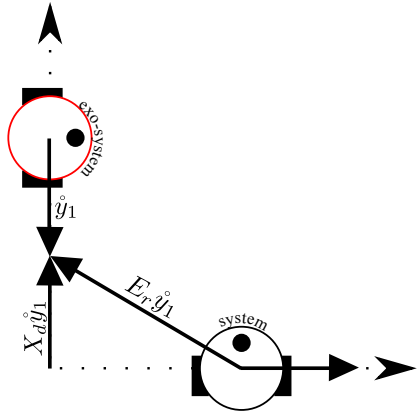}
    \caption{Reference vectors and error vectors in $\SE(2)$.The Inertial reference frame is represented with dotted lines.}
    \label{reference vectors}
\end{center}
\end{figure}

The following assumptions will be used in the next sections and they are instrumental for stability analysis.
\begin{assum} \label{assumption_measures}
There are sufficient independent measurements $y_i$ with $i=1,  \ldots ,\nu$ such that
\be \label{ll}
\cost(\EE):= \dfrac{1}{2}\ssum{i=1}{\nu}\dnorm{\ee}_i^2=\dfrac{1}{2}\ssum{i=1}{\nu}\dnorm{y_i^r-\EE_r y_i^r}^2
\eeq
is locally positive definite in $\EE_r \in \grpG$ around the identity matrix $\EE_r=\Id$.
\end{assum}
\begin{assum} \label{assumption_Compact}
There exist a compact set $\mathcal{W}_d \subset \grpG \times \gothg$ which is invariant for \eqref{ExoSystem}.
\end{assum}

\section{A General Output Regulation Result} \label{OutputRegulation}
The main goal of this section is to present a general structure of the regulator that solves the problem of output regulation formulated above.
\begin{thm} \label{MainTh}
Consider system \eqref{System} along with exosystem \eqref{ExoSystem}.  Let the controller be given by the control law
\begin{subequations} \label{eqTheorem}
\begin{align}
U &=Ad_{X^{-1}} {\Delta}- k_{p} \ssum{i=1}{\nu} \pr \left(\ee_i (y_i^{r}-e_i)^\top  \right) \label{U_TH}
\\
\Delta & = \mrp(C\delta)  \label{Delta_TH}
\\[0.5em]
\dot \delta &=S \delta+C^{\top} \, Q \vrp(\beta) \label{delta_TH}
\\
\beta &=-k_I\ssum{i=1}{\nu} \pr(X^{-\top} \ee_i (y_i^{r}-e_i)^\top X^{\top}) \label{beta_TH}
\end{align}
\end{subequations}
with $k_{p}$ and $k_I$ some positive gains.
If Assumptions \ref{assumption_measures} and \ref{assumption_Compact} hold then the compact set
\begin{multline*}
\! \mathcal{S}=\{(X,\delta,(X_d,w)) \in \grpG \times \R^m \times \mathcal{W}_d \>:\\
\> X^{-\!1}\!X_d=X_r, \> \delta=w\}
\end{multline*}
is locally asymptotically stable for the closed-loop system.
Furthermore
\be\label{eInf}
(X,\delta,(X_d,w)) \in \mathcal{S}\Rightarrow e_i=0 \quad \forall i=1,\ldots,\nu \> .\\[0.5em]
\eeq
\end{thm}
Note that in the specific case of exosystem constant velocity, namely $S=0, \> C=I$, the control law proposed in \eqref{eqTheorem} is reduced to the PI control action presented in \cite{ACC15} Theorem 3.1.

\begin{rem}
Local properties of Theorem \ref{MainTh} can be extended to global ones exploiting the particular structure of the Lie Group considered. We will show in the next section that, for the specific case of systems posed on the special orthogonal group $SO(3)$, the control law in \eqref{eqTheorem} achieve almost global stability of the error system dynamics. In the special orthogonal group of rotation there is a topological obstruction that deny the possibility to have global results with a smooth control action \cite{Obstruction}. To overcame this topological constraint an hybrid control law is needed \cite{Hybrid}.
\end{rem}
\begin{proof} It's straightforward  to verify that condition \eqref{eInf} follows from the definition of the set $\mathcal{S}$ and of $\ee_i$'s.
In what follow we prove that $\mathcal{S}$ is locally asymptotically stable.
Consider the following candidate Lyapunov function
\be\label{LyapCandidate}
\mathcal{L}=\underbrace{ \dfrac{1}{2} \ssum {i=1}{\nu} \dnorm {\ee_i}^2}_{\mathcal{L}_1}+ \underbrace{ \dfrac{1}{2k_I} \tilde{w}^{\top}\tilde{w}}_{\mathcal{L}_2}
\eeq
where $\tilde{w}=w-\delta$, which, by assumption,  is positive definite around $\EE_r=\Id$.
Taking the derivatives along the solution of \eqref{System} and \eqref{ExoSystem} one obtains
\be\label{Lderivative}
\dot{\mathcal{L}}=\dot{\mathcal{L}_1}+\dot{\mathcal{L}_2}.
\eeq
The first element of the equation above can be written as
\be\label{L1derivative}
\dot{\mathcal{L}_1}= \dfrac{1}{2} \ssum{i=1}{\nu} \dfrac{d}{dt}\dnorm{y_i^r-\EE_ry_i^r}^2
=-\ssum{i=1}{\nu}e_i^{\top}\dot E_r y_i^r \> .
\eeq
Bearing in mind the definition of $E_r$ in \eqref{stateError}, it turns out that the time derivative of $E_r$ is given by
\be\ba{lll} \label{ErrorDynamic}
\dot \EE_r &= \dfrac{d}{dt}(X^{-1} X_dX_r^{-1})
\\
&=(X^{-1}(- \dot X) X^{-1})X_d X_r^{-1}+X^{-1}\dot {X_d} X_r^{-1}
\\
&=-X^{-1}XUX^{-1}X_dX_r^{-1}-X^{-1} {}^\circ U_d X_dX_r^{-1}
\\
&=-(UX^{-1}-X^{-1}{}^\circ U_d XX^{-1})X_d X_r^{-1}
\\
&=-(U-Ad_{X^{-1}}{}^\circ U_d )\EE_r
\ea\eeq
and, substituting $\dot{\EE_r}$ into \eqref{L1derivative}, one obtains
\be\ba{rcl}\label{L1Adj}
\dot{\mathcal{L}_1}&=&\ssum{i=1}{\nu}\ee_i^{\top}(U-Ad_{X^{-1}}{}^\circ U_d )\EE_r y_i^r\\
&=&\ssum{i=1}{\nu}\ee_i^{\top}(U-Ad_{X^{-1}}\tilde{\Delta}-Ad_{X^{-1}}{\Delta})\EE_r y_i^r
\ea\eeq
where
\[
\tilde{\Delta}={}^\circ U_d -{\Delta} \> .
\]
Now let's focus on the \textit{i}'th term of \eqref{L1Adj}. It follows \\[0.5em]
\[\ba{rcl}
\ee_i^{\top}\left(U-Ad_{X^{-1}}\tilde{\Delta}-Ad_{X^{-1}}{\Delta}\right)\EE_r y_i^r=\trace{\ee_i^{\top}\left(U-Ad_{X^{-1}}\tilde{\Delta}-Ad_{X^{-1}}{\Delta}\right)\EE_r y_i^r}
\ea\]


\begin{multline*}
=\trace{\left(U-Ad_{X^{-1}}{\Delta}\right)^{\top}\pr \left( \left(\EE_r y_i^{r}  \ee_i^\top\right)^\top\right)}
-\trace{\tilde{\Delta}^{\top}  \pr \left( X^{-\top} \left(\EE_r y_i^{r}  \ee_i^\top\right)^\top X^{\top}\right)}
\end{multline*}
where in the last equation above it has been introduced the projection $\pr$ associated with the Lie-algebra.
Recalling $\dot{\mathcal{L}}_1$ one has
\begin{multline} \label{L1dotFinal}
\dot{\mathcal{L}_1}=\ssum{i=1}{\nu}\trace{\left(U-Ad_{X^{-1}}{\Delta}\right)^{\top}\pr \left( \left(\EE_r y_i^{r}  \ee_i^\top\right)^\top\right)}\\
-\ssum{i=1}{\nu} \trace{\tilde{\Delta}^{\top}  \pr \left( X^{-\top} \left(\EE_r y_i^{r}  \ee_i^\top\right)^\top X^{\top}\right)}.
\end{multline}
The second term of \eqref{Lderivative} can be written as
\[
\dot{\mathcal{L}_2}= \displaystyle\frac{1}{k_I} {\tilde{w}}^{\top} \dot{\tilde{w}}\\[0.6em]
=\displaystyle \frac{1}{k_I} {\tilde{w}}^{\top}(Sw-\dot \delta)\\[0.6em]
\]
and, substituting $\dot{\delta}$ from \eqref{delta_TH} into the above equation, one obtains
\[\ba{lll}
\displaystyle \frac{1}{k_I}  {\tilde{w}}^{\top}(Sw-\dot \delta)=\displaystyle \frac{1}{k_I} {\tilde{w}}^{\top}S\tilde{w}-\displaystyle \frac{1}{k_I} {\tilde{w}}^{\top} C^{\top} Q \> \vrp(\beta).
\ea\]
Bearing in mind that $S$ is a skew-symmetric matrix and hence ${\tilde{w}}^{\top}S\tilde{w}=0$,  one gets
\[\ba{l}
\displaystyle \frac{1}{k_I} {\tilde{w}}^{\top}S\tilde{w}-\frac{1}{k_I} {\tilde{w}}^{\top} C^{\top} Q \> \vrp(\beta)
\\
\ba{rcl}
\qquad \qquad \qquad&=&-\displaystyle \frac{1}{k_I} (C{\tilde{w}})^{\top} Q \vrp(\beta)\\[0.6em]
&=&-\displaystyle \frac{1}{k_I}\vrp^\top( {\tilde{\Delta}}) Q	\vrp(\beta)\\[0.6em]
&=&-\displaystyle \frac{1}{k_I}tr( {\tilde{\Delta}}^{\top} \beta) \>.
\ea\ea\]
and, recalling the expression of $\dot{\mathcal{L}}_2$ and the ``innovation term" $\beta$ \eqref{beta_TH}, one has
\be\label{L2derivative}
\dot{\mathcal{L}_2}=\ssum{i=1}{\nu}\trace{\tilde{\Delta}^{\top} \pr \left(X^{-\top} \left(\EE_r y_i^{r}  \ee_i^\top\right)^\top  X^{\top}\right)}.
\eeq
Finally, substituting \eqref{L1dotFinal} and \eqref{L2derivative} into \eqref{Lderivative} and introducing the expression of $U$ \eqref{U_TH}, it yields
\[\ba{rcl}
\dot{\mathcal{L}}&=&\ssum{i=1}{\nu}\trace{\left(U-Ad_{X^{-1}}{\Delta}\right)^{\top}\pr \left( \left(\EE_r y_i^{r}  \ee_i^\top\right)^\top\right)}
\\
&=&\ssum{i=1}{\nu}\trace{-k_{p} \pr \left( \left(\EE_r y_i^{r}  \ee_i^\top\right)^\top\right)^\top \pr \left( \left(\EE_r y_i^{r}  \ee_i^\top\right)^\top\right)}
\\
&=&- k_{p} \ssum{i=1}{\nu} \norm{ \pr \left( \left(\EE_r y_i^{r}  \ee_i^\top\right)^\top\right)}^2.
\ea\]
Since $\mathcal{L}$ is positive definite in the error state and since the exosystem state $(X_d,w)$ lies in a compact set, it follows that the whole state is globally bounded and solutions exist for all time. From this, 
using similar LaSalle arguments to Theorem 3.1 in \cite{ACC15}, it is possible to conclude that the set $\mathcal{S}$ is locally asymptotically stable.
\end{proof}

\section{The Special Case of $\SO(3)$}\label{SO3case}
In this section we consider the specific case of systems defined in $\SO(3)$. By changing the notation used in the previous sections to a more classical notation we denote by $R \in \SO(3)$ the rotation matrix of the system and by $\Omega \in \R^3$ its angular velocity. 
In this section we consider the case in which the kynematic system \eqref{System} is completed with the dynamics equation of motion
\begin{subequations} 
\begin{align}
\dot{R}&=R \Omega_\times  \label{SO3system}
\\
J\dot{\Omega}&=- \Omega \wedge J \Omega + \Gamma. \label{SO3systemDynamic}
\end{align}
\end{subequations} 
Thus modelling a rigid body with control input $\Gamma \in \RR^3$ and $J$ the inertia matrix. In this context the rotation velocity $\Omega$ is a state component of the system.
In this framework the exosystem is described by
\be\ba{rcl}\label{SO3exo-system}
\dot{R}_d&=&{}^\circ \Omega_{d\times} R_d
\\
{}^\circ \Omega_{d\times}&=&(Cw)_\times
\\
\dot{w}&=&Sw
\ea\eeq
where $R_d$ is the desired orientation and ${}^\circ \Omega_d$ is the desired angular velocity in the inertial reference frame.
In order to deal with the new control input $\Gamma$, we split the problem into two sub-problems following the backstepping paradigm \cite{Khalil}.
In the first sub-problem, Section \ref{OmegaVirtual}, we consider $\Omega$ as control input instead of a state component of the system and we solve the output regulation problem providing the stability analysis.
In the second one,  Section \ref{BS}, we consider the control input $\Omega\equiv\Omega^c$ as virtual velocity reference.  The virtual control input $\Omega^c$, then,  is backstepped in order to get the real torque input $\Gamma$.

\subsection{Angular velocity $\Omega_\times$ as control input}\label{OmegaVirtual}
In this context it turns out that the control law \eqref{eqTheorem} reads as (with $R_e:=E_r$) 
\begin{subequations} \label{eqPropositionSO3}
\begin{align}
\Omega_{\times}^c &=Ad_{R^\top} {\Delta}_\times+\dfrac{k_{p}}{2} \ssum{i=1}{\nu}  \left(e_i\wedge y_i^{r} \right)_\times \label{Omega_CoSO3}
\\
\Delta_\times & = (C\delta)_\times  \label{Delta_CoSO3}
\\[0.5em]
\dot \delta &=S \delta+C^{\top} \, Q_{\so(3)} \beta \label{delta_CoSO3}
\\
\beta &=\dfrac{k_I}{2} R\ssum{i=1}{\nu} \left(e_i \wedge y_i^{r}\right) \, \label{beta_CoSO3}
\end{align}
\end{subequations}
where $\Omega^c$ is the desired virtual input mentioned earlier.
We continue the analysis under the following assumption.
\begin{assum} \label{Measurement in SO3}
There are two or more linear independent measurements $y_i$ ($i=1,\ldots,\nu; \, \nu \geq 2$) and
the symmetric matrix $$Y=\dfrac{k_{p}}{2} \ssum{i=1}{\nu}  y_i^r y_i^{r^\top }$$ has three distinct eigenvalues.
\end{assum}
Note that Assumption \ref{Measurement in SO3} specializes Assumption \ref{assumption_measures} in the specific case of systems posed on $SO(3)$. Indeed, it is straightforward to verify that two is the minimum number of measurements such that the cost function \eqref{ll} is positive definite around the identity element of the group.

Theorem \ref{MainTh} ensures the local attractiveness of the compact set S in \eqref{SO3set}, exploiting the structure of the special orthogonal group one can shown  further properties of the set, namely almost global stability and local exponential stability. The extensions of the local properties of Theorem \ref{MainTh} in the special case of $\SO(3)$ are stated in the following Proposition. 
\begin{prop} \label{propositionSO3}
Consider the system \eqref{SO3system} along with exosystem \eqref{SO3exo-system} and let the controller be given by \eqref{eqPropositionSO3}.
Let Assumption \ref{Measurement in SO3} holds. Then the set
\begin{multline}\label{SO3set}
\! \mathcal{S}=\{(R,\delta,(R_d, w)) \in \SO(3) \times \R^m \times (\SO(3) \times \R^m) \>:\\
\> R^\top \!R_d=X_r, \> \delta=w\}
\end{multline}
is  almost globally asymptotically and locally exponentially stable for the closed-loop system. 
\end{prop}

Note that Assumption \ref{assumption_Compact} is no more needed in Proposition \ref{propositionSO3} since the special orthogonal group $\SO(3)$ is a compact manifold.
\begin{proof}
It is straightforward to verify that, under Assumption \ref{Measurement in SO3}, the Lyapunov function \eqref{LyapCandidate} is positive definite and $\mathcal{L}(0,0)=0$.

 In what follows we proceed by steps showing that:
 \begin{enumerate}
 \item The dynamic of the group error for the closed-loop system has only four isolated equilibrium points $(R_e, \Delta)=(R_{ej}^*,0)$, $j=1,\ldots,4$.
 \item The equilibrium point $(R_{e1}^*,\Delta)=(I_3,0)$ is locally exponentially stable.
 \item The three equilibria with $(R_{ej}^*,\Delta)\neq (I_3,0)$, $j=2,\ldots,4$ are unstable.
 \end{enumerate}
 Consider the dynamic of the group error for the closed-loop system
\[
 \dot{R}_e=\left( Ad_{R^\top} \tilde{\Delta}_\times+ \dfrac{k_{p}}{2} \ssum{i=1}{\nu}  \left(R_e y_i^{r} \wedge y_i^{r}\right)_\times \right)R_e\,.
 \]
 As consequence of Theorem \ref{MainTh} for $\dot{R}_e=0$ one has that $\tilde{\Delta}=0$, which implies 
 \be\ba{rcl} \label{SO3R_eEquilib}
0&=& \dfrac{k_{p}}{2}  \ssum{i=1}{\nu} \left(R_e^* y_i^{r} \wedge y_i^{r}\right)\\
&=&R_e^* \dfrac{k_{p}}{2} \ssum{i=1}{\nu} y_i^{r} y_i^{r^\top}-\dfrac{k_{p}}{2} \ssum{i=1}{\nu} y_i^{r} y_i^{r^\top} R_e^{*^\top}
=R_e^*Y-YR_e^{*^\top}\!\!.
\ea\eeq
Proceeding like in the proof of Theorem 5.1 in \cite{SO3TarekMahony} one has that $R_e^*Y=YR_e^{*^\top}$ implies either $R_e^*=I_3$ or $tr(R_e^*)=-1$. 
As consequence, $R_e^*$ is a symmetric matrix and there are only four possible values of $R_e^*$ that satisfy eq. \eqref{SO3R_eEquilib}, namely
\[\left\{ 
  \begin{array}{l l}
    R_{e1}^*=I_3&\\
    R_{e2}^*=u_1 u_1^\top-u_2 u_2^\top-u_3 u_3^\top &\\
    R_{e3}^*=-u_1 u_1^\top+u_2 u_2^\top-u_3 u_3^\top &\\
    R_{e4}^*=-u_1 u_1^\top-u_2 u_2^\top+u_3 u_3^\top &
  \end{array} \right.\]
  where $u_1,\,  u_2 , \, u_3$ are the eigenvectors of $Y$ associated to the eigenvalues $\lambda_1, \lambda_2, \lambda_3$ ordered in ascending order. This, in turn,  proves  step 1.
Hereafter we prove item 2 (namely local exponential stability of the set $\mathcal{S}$).
Without loss of generality consider $X_r=I_3, \, C=I_3$. The dynamics of the error and velocity error for the closed-loop system can be written as 
\[\ba{rcl}
\dot{R}_e&=&\left( R_e R_d^\top \tilde{\Delta}+\dfrac{k_{p}}{2} \ssum{i=1}{\nu}  \left(R_e y_i^{r} \wedge y_i^{r}\right)\right)_\times \!\!\! R_e
\\
\dot{\tilde{\Delta}}_\times &=& \left( S \tilde{\Delta} + k_{I} \ssum{i=1}{\nu}  R_d R_e^\top \left(R_e y_i^{r} \wedge y_i^{r}\right) \right)_\times
\ea\]
and, denoting
\[
\underline{\tilde{\Delta}}=\underline{R} \tilde{\Delta}\,,
\qquad 
\underline{\dot{R}}=-\underline{R}S
\]
where $\underline{R} \in \SO(3)$, one gets
\be\ba{rcl} \label{SO3-EquivalentErrorDynamics}
\dot{R}_e&=&\left( R_e R_d^\top \underline{R}^\top \tilde{\underline{\Delta}}+\dfrac{k_{p}}{2} \ssum{i=1}{\nu}  \left(R_e y_i^{r} \wedge y_i^{r}\right)\right)_\times \!\!\! R_e
\\
\dot{\tilde{\underline{\Delta}}} &=&k_{I} \ssum{i=1}{\nu} \underline{R} R_d R_e^\top \left(R_e y_i^{r} \wedge y_i^{r}\right) \, .
\ea\eeq
Consider a first order approximation of $(R_e, \underline{\tilde{\Delta}})$ of equation  \eqref{SO3-EquivalentErrorDynamics} around the equilibrium point $(R_{e1}^*,0)$. It suffices to show that the origin of the linearized system is uniformly exponentially stable in order to prove that the equilibrium $(R_{e1}^*,0)$ is locally exponentially stable.
To this purpose consider the approximation $R_e=I+x_\times$ and $\Delta=\theta$, with $x,\theta \in \RR^3$. By neglecting high order terms one obtains
\be \label{Loria_standard_form}
\begin{bmatrix} \dot{x} \\ \dot{\theta} \end{bmatrix}=
\begin{bmatrix}
\mathcal{A}(t) & \mathcal{B}(t)^\top 
\\
-\mathcal{C}(t) & 0
\end{bmatrix}
\begin{bmatrix}
{x} \\ {\theta}
\end{bmatrix}
\eeq 
where
$$\ba{rcl}
\mathcal{A}(t)&:=&\dfrac{k_{p}}{2} \ssum{i=1}{\nu}  y_{i\times}y_{i\times}, \qquad
\mathcal{B}(t) := \underline{R}R_d,\\
\mathcal{C}(t) &:=& -k_{I} \underline{R} R_d \ssum{i=1}{\nu}  y_{i\times}y_{i\times}.
\ea$$
The proof follows the results derived from Theorem 1 in \cite{Loria} which establish sufficient conditions for the uniform exponential stability of the origin of a linear time-varying system having the standard form in \eqref{Loria_standard_form}.
As a matter of fact, the constant symmetric positive definite matrices $\mathcal{P}=k_{I} \sum_{i=1}^{\nu}  y_{i\times}y_{i\times}^\top$ and $\mathcal{Q}=k_I^{-1}k_p \mathcal{P}^2$ satisfy the conditions $\mathcal{P} \mathcal{B}^\top=\mathcal{C}^\top$ and $-\mathcal{Q}=\mathcal{A}^\top\mathcal{P}+\mathcal{P}\mathcal{A}+\dot{\mathcal{P}}$ in Assumption 2 of  \cite{Loria}. Furthermore,  it is straightforward to verify that the term $\mathcal{B}$ is uniformly persistently exciting, namely  for any positive constant $\varepsilon$ there exists $T>0$ such that
$$\int_t^{t+T} \!\!\!\!\!\! \mathcal{B}(\tau)  \mathcal{B}(\tau)^\top d\tau= \int_t^{t+T}\!\!\!\!\!\! \underline{R} R_d R_d^\top \underline{R}^\top d\tau=TI_3>\varepsilon I_3.$$
Moreover from Theorem \ref{MainTh} one has that  $\norm{\mathcal{B}(t)}$ and \scalebox{0.65}{$\norm{\dfrac{\partial \mathcal{B}(t)}{\partial t}}$} remain bounded for all time.
Hence  all conditions of Theorem 1 in \cite{Loria} are satisfied, thus the claim in item 2 holds.
In order to prove that the others three equilibria $(R_{ej}^*,0)$, $j \in \{2,3,4\}$, are unstable, it suffices to prove that the origin of the linearized time-varying system is unstable.
To this purpose consider the first order approximation $R_e=R_{ej}^*(I_3+x_\times)$, $\Delta=\theta$ with $x, \,y \in \R^3$ around an equilibrium point $(R_{ej}^*, 0)$.
The linearization of \eqref{SO3-EquivalentErrorDynamics} is given by
\be\ba{rcl} \label{SO3_linearization_unstable}
\dot{x}_j&=&\Upsilon_j x_j+R_d^\top \underline{R}^\top \theta
\\[0.6em]
\dot{\theta}_j&=&2\dfrac{k_I}{k_p} \,\underline{R} R_d \Upsilon_j x_j
\ea\eeq
where $\Upsilon_{j} :=\dfrac{k_{p}}{2}  \ssum{i=1}{\nu} R_{ej}^* y_{i\times} R_{ej}^* y_{i\times}$, $\mbox{with} \> j=2,\ldots,4$.

The proof follows results derived from the Chataev's theorem \cite{Khalil}. In particular,
consider the functions
\[
\mathcal{V}_j(x_j,\theta_j)= \dfrac{k_I}{2k_p} x_j^\top \Upsilon_{j}^\top x_j -\dfrac{1}{4} \norm{\theta_j}^2,\quad j=2,\ldots,4
\]
and, for an arbitrarily small radius $r>0$,  define the set
$$\mathcal{U}_{j,r}:=\{(x_j,\theta_j)^\top  \mid \mathcal{V}_j(x_j,\theta_j)>0, \norm{x_j,\theta_j}<r \}.$$
Following the Chateav's theorem, we show now that $\Upsilon_j$ is not singular, at least one of its eigenvalues is positive and that $\mathcal{U}_{j,r}$ is non-empty  for each $j \in \{2,3,4\}$ and $r>0$.
As a matter of fact, consider the characteristic polynomial of the matrix $\Upsilon_j$.  One has
 \be\ba{rcl} \label{SO3_eigenvalues}
 \det(\Upsilon_j-\bar{\lambda}I_3)&=& \det(YR_{ej}^*-k_{p} \ssum{i=1}{\nu}  y_i^\top R_{ej}^* y_i I_3-\bar{\lambda}I_3)
 \\[0.5em]
 &=& \det(R_q  \lambda_q R_q^\top R_q \bar{R}_{j}  R_q^\top- \trace{\lambda_q  \bar{R}_{j}}I_3-\underline{\lambda} I_3)
\\[0.5em]
& = &\det(\lambda_q \bar{R}_{j}-\trace{\lambda_q \bar{R}_{j}}I_3-\underline{\lambda} I_3)
\\[0.5em]
 &=&\det(\lambda_q \bar{R}_{j}-\trace{\lambda_q \bar{R}_{j}}I_3-\bar{\lambda}I_3).
 \ea\eeq
where we have used the fact that the matrices $Y$ and $R_{ej}^*$ can be decomposed as $Y=R_q \lambda_q R_q^\top$ and $R_{ej}^*=R_q \bar{R}_{j} R_q^\top$ with $\lambda_q=\diag(\lambda_1,\lambda_2,\lambda_3)$ and
\[\ba{rcl}
\bar{R}_2=\diag(1,-1,-1), \quad \bar{R}_3=\diag(-1,1,-1), \quad \bar{R}_4=\diag(-1,-1,1).
\ea\]
From equation \eqref{SO3_eigenvalues} one gets
\[\ba{rcl}
\lambda(\Upsilon_2)&=&[\lambda_2+\lambda_3; \lambda_3-\lambda_1; \lambda_2 - \lambda_1]^\top\\
 \lambda(\Upsilon_3)&=&[\lambda_3-\lambda_2; \lambda_3+\lambda_1;  \lambda_1-\lambda_2]^\top\\
 \lambda(\Upsilon_4)&=&[\lambda_2-\lambda_3; \lambda_1-\lambda_3;  \lambda_1+\lambda_2]^\top.
 \ea\]
 Since,  by Assumption \ref{Measurement in SO3}, the eigenvalues $\lambda_1$, $\lambda_2$, $\lambda_3$ are distinct and 
 $$0 \leq \lambda_1 < \lambda_2 < \lambda_3$$ 
 it is straightforward to verify that the matrix $\Upsilon_j$ is not singular and at least one of its eigenvalues is positive.  As consequence the set $\mathcal{U}_{j,r}$ is non-empty. 
Consider now the derivatives of  $\mathcal{V}_j$ along the trajectories of the system. One has
\[
\dot{\mathcal{V}}_j (x_j, \theta_j)= \dfrac{k_I}{k_p} x_j^\top \Upsilon_j^\top \Upsilon_j x_j.
\]
Since the matrix $\Upsilon_j$ is full rank it is straightforward to verify that $\Upsilon_j^\top \Upsilon_j>0$, $\mbox{with} \> j=2,\ldots,4$. As consequence $\dot{\mathcal{V}}_j$ is positive for each $(x,\theta) \in \mathcal{U}_{j,r}$.
From this  we conclude that all the conditions of Chateav's Theorem are satisfied and thus the origin of the system \eqref{SO3_linearization_unstable} is unstable for $j \in \{2,3,4\}$.  This completes the proof.
\end{proof}
\subsection{Backstepping procedure}\label{BS}
In order to deal with the new control input $\Gamma$ starting with the virtual input $\Omega^c$, a backstepping procedure is developed.
Define $\tilde{\Omega}$ as 
$$\tilde{\Omega}=\Omega-\Omega^c.$$
Consider \eqref{Delta_TH}, \eqref{delta_TH} completed with the following control law
\begin{subequations} \label{SO3-BS_CTRL_LAW}
\begin{align}
\Gamma & = \label{SO3-BS-Gamma}
	\begin{multlined}[t]
		\Omega_\times J \Omega 
		-J \Omega_\times R^\top \Delta
		+J R^\top \dot{\Delta}
		+2 \ssum{i=1}{\nu}\left(e_i \wedge y_i^{r}\right) 
	\\ 
		+ Jk_{p} \ssum{i=1}{\nu}  y_{i\times}^r \left(y_i^r - e_i \right)_\times  (\tilde{\Omega}+\alpha)
		-k_D \tilde{\Omega}
  \end{multlined} \\
\beta &=\label{SO3-BS-Beta}
\begin{multlined}[t]
		\!\! \dfrac{k_I}{2} \left( \ssum{i=1}{\nu}\left( k_{p} R \left( y_i^r-e_i\right)_\times  y_{i\times}^r  J^\top \tilde{\Omega}  +
		R \left(e_i \wedge y_i^{r}\right)\right) \right)		
  \end{multlined}
\end{align}
\end{subequations}
with $\alpha=k_{p}\sum_{i=1}^{\nu} 0.5 \, (e_i \wedge y_i^r)$ and $k_D$ a positive arbitrary gain.
By backstepping $\tilde{\Omega}$ it turns out that control law \eqref{SO3-BS_CTRL_LAW} solves the stabilization problem as stated in the following proposition.
\begin{prop}\label{Proposition_SO3_BS}
Consider system \eqref{SO3system}, \eqref{SO3systemDynamic} along with exosystem \eqref{SO3exo-system} and let the controller be given by \eqref{SO3-BS_CTRL_LAW}, \eqref{Delta_TH}, \eqref{delta_TH}. Let Assumption \ref{Measurement in SO3} holds.Then the set
\begin{multline*}
\!\!\!\!\!\! \mathcal{S}_1=\{((R,R_d), (\delta, w), (\Omega_\times, {}^\circ \Omega_{d\times})) :\\
\in \SO(3)^2 \times \R^{2m} \times \so(3)^2 \>: \> R^\top \!R_d=X_r, \> \delta=w\ , \Omega={}^\circ \Omega_d\}
\end{multline*}
is almost globally asymptotically stable and locally exponentially stable for the closed-loop system.
\end{prop}
\begin{proof}
Consider the following function
$$\mathcal{L}_{bs}=\mathcal{L}+\dfrac{1}{2}\tilde{\Omega}^\top J \tilde{\Omega}$$
where $\mathcal{L}$ is the Lyapunov function defined in \eqref{LyapCandidate}.  It is straightforward to verify that under assumption \ref{Measurement in SO3} the function $\mathcal{L}_{bs}$ is definite positive respect to the compact set $S_1$. 
Differentiating $\mathcal{L}_{bs}$ along the solutions of the closed-loop system and bearing in mind the expression of $\Omega^c$, one has
\[\ba{rcl}
\dot{\mathcal{L}}_{bs} &=&- \dfrac{k_{p}}{4} \ssum{i=1}{\nu}  \norm{ \left(R_e y_i^{r} \wedge y_i^{r}\right)_\times}^2
+ \dfrac{1}{k_I} \tilde{w}^\top \left(Sw-\dot{\delta}\right) 
\\[0.7em]
&&+ \dfrac{1}{2} tr \left( \tilde{\Omega}_\times^\top  \ssum{i=1}{\nu}\left(R_e y_i^{r} \wedge y_i^{r}\right)_\times \right)
\\[0.7em]
&&-\dfrac{1}{2} tr \left( \tilde{\Delta}_\times^\top Ad_{R} \ssum{i=1}{\nu} \left(R_e y_i^{r} \wedge y_i^{r}\right)_\times \right) + \tilde{\Omega}^\top \left( -\Omega_\times J \Omega + \Gamma - J \dot{\Omega^c} \right).
\ea\]
Recalling the expression of $\Omega_\times^c$ in vectorial form and differentiating it, one obtains
\[\ba{rcl} 
\dot{\Omega}^c&=&
- \Omega_\times R^\top \Delta 
+ R^\top \dot{\Delta}
+\dfrac{k_{p}}{2} \ssum{i=1}{\nu}  y_{i\times}^r \left(R_e y_i^r\right)_\times \left(\tilde{\Omega}+\alpha \right)
\\
&&-\dfrac{k_{p}}{2} \ssum{i=1}{\nu}  y_{i\times}^r \left(R_e y_i^r\right)_\times R^\top \tilde{\Delta}.
\ea\]
Substituting the expression of $\dot{\Omega}^c$ in the Lyapunov function and recalling the expression of $\Gamma$ \eqref{SO3-BS-Gamma} and $\dot{\delta}$ \eqref{delta_TH}, one has
\[\ba{rcl}
\dot{\mathcal{L}}_{bs} &=& -\dfrac{k_{p}}{4} \ssum{i=1}{\nu}  \norm{ \left(R_e y_i^{r} \wedge y_i^{r}\right)_\times}^2
-k_D \tilde{\Omega}^2

\\
&&  + \tilde{\Delta}^\top \Bigg(k_{p} \ssum{i=1}{\nu} R \left(R_e y_i^r\right)_\times  y_{i\times}^r  J^\top \tilde{\Omega} -R \ssum{i=1}{\nu} \left(R_e y_i^{r} \wedge y_i^{r}\right)- \dfrac{2}{k_I} \beta \Bigg).

\ea\]
%
%
%
%
%
%
Introducing the expression of $\beta$ \eqref{SO3-BS-Beta} in the  above expression one obtains
\[
\dot{\mathcal{L}}_{bs} =- \dfrac{k_{p}}{4} \ssum{i=1}{\nu}  \norm{ \left(R_e y_i^{r} \wedge y_i^{r}\right)_\times}^2
-k_D \tilde{\Omega}^2.
\]
It follows that the compact set $S_1$ is stable in the sense of Lyapunov and that $\tilde{\Omega}$ converges to zero.
The proof can be completed using similar arguments to Proposition \ref{propositionSO3}.
\end{proof}
\section{Simulation Results}\label{Simulation Results}
As illustrative example we consider two rigid bodies posed on $\SO(3)$. The first one, the simulated system, is modeled as fully actuated dynamic system \eqref{SO3systemDynamic} while the second one, the simulated exosystem, is modeled as kinematic system \eqref{SO3exo-system}.
The reference directions considered in the simulations are $\mathring y_1=[1, 0, 0]^\top$ and $\mathring y_2=[0,1,0]^\top$.
Initial states of the simulated system are chosen as $R(0)=I_3$, $\Omega(0)=0$ while the controller gains are chosen to be $k_p=k_d=2$ and $k_I=0.4$.
The inertia matrix in the body-fixed frame of the system is that of an non-axisymmetric rigid body $J = \diag(2, 1.5, 1)$[kg m$^2$].
The desired velocity in the inertial frame was chosen as ${}^\circ\Omega_d=(\cos(t), 2\cos(5t),3\cos(7t))\> [\mbox{rad s}^{-1}]$.
The initial yaw, pitch and roll of the simulated exosystem are chosen respectively as $180^\circ$, $45^\circ$ and $45^\circ$.
Figure \ref{so3error} and Figure \ref{so3errorZoomed} show the evolution of $\trace{I-R_e}$ and the relative velocity error $\norm{\tilde{\Omega}}$.
Plots show that, in steady state, the velocity of the simulated system $\Omega$ converges to $\Omega^c$ and the relative error $R_e$ converges to the identity element of the group.
\begin{figure}[H]
\begin{center}
    \includegraphics[trim = 13mm 0mm 0mm 0mm, width=.65\textwidth]{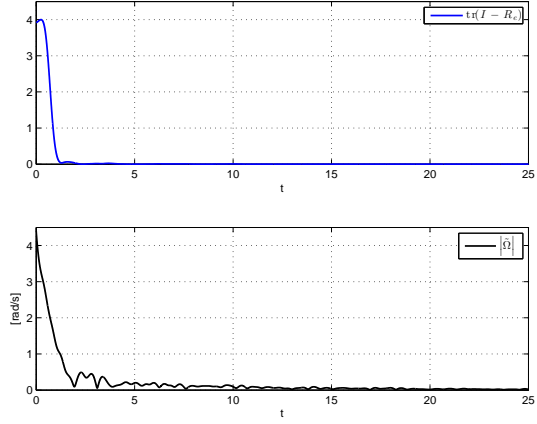}
    \caption{Evolution of the relative Error $R_e$ and velocity error $\tilde{\Omega}$ in $\SO(3)$. Case of perfect knowledge of the Inertia matrix.}
    \label{so3error}
\end{center}
\end{figure}
\begin{figure}[H]
\begin{center}
    \includegraphics[trim = 13mm 0mm 0mm 0mm, width=.65\textwidth]{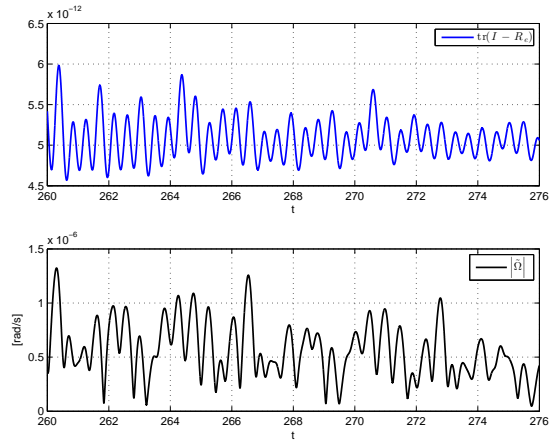}
    \caption{Evolution of the relative Error $R_e$ and velocity error $\tilde{\Omega}$ in $\SO(3)$ after 260 seconds. Case of perfect knowledge of the Inertia matrix.}
    \label{so3errorZoomed}
\end{center}
\end{figure}

A second simulation was run,  considering a slight unknown variation in the inertia tensor, to have a numerical hint on the robustness properties of the proposed control law.
The inertia matrix of the implemented control law is chosen as $J_{nom} = \diag(2, 1.5, 1)$[kg m$^2$] while the real value is $J_{real} = \diag(2.2, 1.3, 1.1)$[kg m$^2$].
Figure \ref{so3errorNoRobust} and Figure \ref{so3errorNoRobustZoom} show that only practical regulation can be achieved in presence of uncertainties in the inertia tensor.
\begin{figure}[H]
\begin{center}
    \includegraphics[trim = 13mm 0mm 0mm 0mm, width=.65\textwidth]{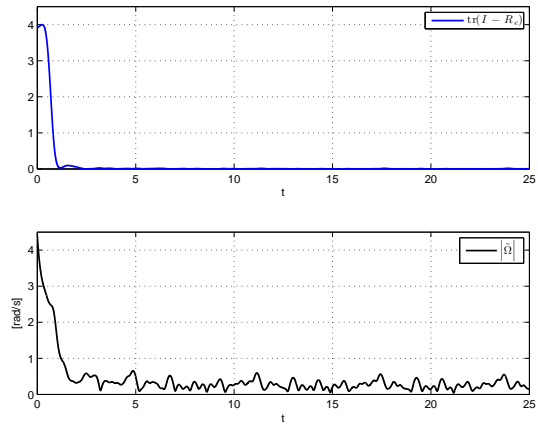}
    \caption{Evolution of the relative Error $R_e$ and velocity error $\tilde{\Omega}$ in $\SO(3)$. Case of uncertainties in the Inertia.}
    \label{so3errorNoRobust}
\end{center}
\end{figure}
\begin{figure}[H]
\begin{center}
    \includegraphics[trim = 13mm 0mm 0mm 0mm, width=.65\textwidth]{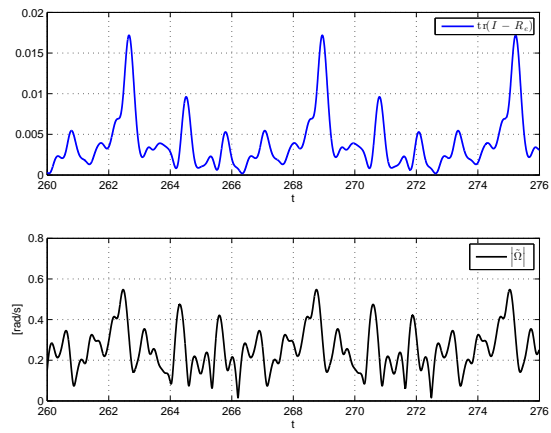}
    \caption{Evolution of the relative Error $R_e$ and velocity error $\tilde{\Omega}$ in $\SO(3)$ after 260 second. Case of uncertainties in the Inertia.}
    \label{so3errorNoRobustZoom}
\end{center}
\end{figure}

\section{Conclusion}
In this paper, the problem of output regulation for systems defined on matrix Lie-Groups with invariant measurements was considered. 
The proposed control law structure embeds a copy of the exosystem kinematic updated by means of error measurements.
A rigorous stability analysis was provided for both the general case and the particular case of systems posed on the special orthogonal group $\SO(3)$.
Future works will consider the problem of robustness with respect to uncertainties in the parameters (for instance, in the inertia tensor) and external unknown disturbances. Note, in fact, that the control law proposed in \eqref{SO3-BS_CTRL_LAW} assumes a perfect knowledge of them.

\section*{Acknowledgement}                             
This research was supported by integrated project SHERPA (G.A. 600958) supported by the European Community under the 7th Framework Programme, and by the Australian Research Council and MadJInnovation through the Linkage grant LP110200768.

\bibliographystyle{plain}        
\bibliography{biblio}           
                                 

\end{document}